\documentclass[10pt]{amsart}
\usepackage{amsmath,amsthm}
\usepackage{amssymb}
\usepackage{centernot} 
\usepackage{tikz}

\newtheorem{thm}{Theorem}
\newtheorem{prop}{Proposition}
\newtheorem*{conj}{Conjecture}

\begin{document}

\begin{center}

\textbf{\LARGE{ON SOME SYMMETRY AXIOMS }}
\vspace{8pt}

\textbf{\LARGE{IN RELATIVITY THEORIES}}
\vspace{24pt}

\baselineskip 16pt
\large{Gergely Sz\'ekely}
\end{center}

\vspace{24pt}

\baselineskip 10pt

\noindent
\footnotesize{\textit{Address}: MTA Alfr\'ed R\'enyi Institute of Mathematics,\\ Re\'altanoda utca 13-15, H-1053, Budapest, Hungary.
\\E-mail: szekely.gergely@renyi.mta.hu.}

\vspace{24pt}

\baselineskip 14pt

\normalsize

\noindent
\textbf{Abstract}:
\textit{
  In this paper we review two symmetry axioms of special relativity and their connections to each other together with their role in some famous predictions of relativity theory, such as time dilation, length contraction, and the twin paradox. We also discuss briefly counterparts of these symmetry axioms in general relativity and formulate a conjecture, namely that without them the axioms of general relativity would capture general relativistic spacetimes only up to conformal equivalence.}

\vspace{12pt}



\baselineskip 14pt
\setlength{\parskip}{12pt}

\section{INTRODUCTION}

Special relativity is axiomatic in its spirit from the beginning since Einstein introduced special relativity as a theory of two informal postulates in his famous 1905 paper. One of these postulates was Galileo's principle of relativity (generalized to all physical phenomena), implying that there is no physical phenomenon that distinguishes some inertial observers; the other was the natural assumption that the speed of light is the same in every direction at least in one inertial frame of reference, see Einstein (1905, Sec.~2).  

Even in the non-axiomatic approaches to general relativity, a spacetime is given as a triple of a manifold $M$, a Lorentzian metric $g$, and an energy-momentum tensor $T$ satisfying Einstein's Field Equations and often some extra conditions, such as the causality or the energy conditions. So general relativity is also axiomatic in its spirit because in this context it is natural to think of $M$, $g$, and $T$ as basic concepts and of their basic properties (e.g., $M$ is being a manifold) together with Einstein's Field Equations and the extra conditions as axioms. 

Formal axiomatizations of relativity theories (both special and general) also have an extended literature, see, e.g., Ax (1978),  Benda (2008), Benda (2015), Goldblatt (1987), Guts (1982), Mundy (1986), Pambuccian (2007), Schutz (1981), Szab\'o (2009)  to mention only a few. One of these formal approaches is the one developed by the research team/school of  Hajnal Andr\'eka and Istv\'an N\'emeti. Here we will stay within the main framework of this team. However, comparing the axiomatizations of different frameworks is also an interesting and important research direction, see, e.g., Andr\'eka and N\'emeti (2014), Barrett and Halvorson (2015), Weatherall (2014), Rosenstock et.\ al.\ (2015). 

In this paper, after recalling the main axiomatic framework of the Andr\'eka--N\'emeti school, we review two of the numerous symmetry axioms appearing in this approach. We also review how these symmetry axioms are related to each other and their role in predictions of relativity theory, such as time dilation, length contraction, and the twin paradox. Finally, we discuss the counterparts of these symmetry axioms in general relativity and formulate a conjecture about what happens if we leave out these counterparts from the axiom system of general relativity, see Conjecture \ref{conj} on p.~\pageref{conj}. 

Apart from slight strengthening of some theorems resulting from a natural generalization of  axiom system $SpecRel_0$ and the simple observations of Propositions~\ref{prop-ineq}, \ref{thm-eq} and \ref{prop-c}, all the theorems and propositions of this paper can be found scattered in the cited references. The main aim of this paper is to investigate the roles and intuitive meanings of the two most often used symmetry axioms of the Andr\'eka--N\'emeti school in a way that can easily be understood even by nonspecialists. 

For methodological reasons, we work in a formal axiomatic framework, which among others is beneficial because it forces us to formulate unambiguous basic assumptions with a clear meaning. See Andr\'eka et.\ al.\ (2002, pp.~1245-1252) and Sz\'ekely (2009, Sec.\ 11.) for more methodological details why first-order logic is an ideal logic for developing axiomatic frameworks for relativity theories. 
A practical advantage of using first-order logic is that the theorems can be machine verified, see, e.g., Govindarajalulu et.\ al.\ (2015), and Stannett and N\'emeti (2014). See Friend (2015) for epistemological significance of the Andr\'eka--N\'emeti approach.

To make our formulas easy to read even to non-logicians, we try to use only simple and natural notations, e.g., the logical connective ``implies'' is denoted by $\Longrightarrow$, and logical connective ``and'' is denoted by comma.
Quantifiers ``for all'' and ``exists'' are denoted by the usual symbols   $\forall$ and $\exists$, respectively.      

\section{Axiomatic framework}

\subsection{Basic concepts}

First we fix a set of basic concepts about which we will formulate some basic assumptions (axioms).\footnote{We do not fix the basic concepts once and for all. We fix them just for creating a framework to formulate certain axioms and axiom systems. The axiomatic method also has the flexibility of modifying the basic concepts. Moreover, comparing axiom systems formulated using different basic concepts is an interesting research area, see, e.g., Andr\'eka and N\'emeti (2014).} Here we use the main framework of the Andr\'eka--N\'emeti school. That is, we have two sorts of basic concepts: \emph{Bodies} (things that move) $B$ and \emph{Quantities} (numbers that are used to coordinatize the moving bodies) $Q$. We have two kinds of distinguished bodies \emph{inertial observers} (or inertial coordinate systems) $IOb$ and \emph{light signals} (or photons) $Ph$. To put an algebraic structure on the quantities, we take the usual operations $+$, $\cdot$ and the ordering $<$ as basic concepts. Finally, we connect the physical sort (of bodies) and the mathematical sort (of quantities) by the worldview relation $W$. 

We use the worldview relation $W$ to express how observers associate coordinates to events (i.e., meetings of bodies). This is done by translating basic relation $W(o,b,x,y,z,t)$ to natural language as ``Observer $o$ coordinatizes body $b$ at space location $(x,y,z)$ at instant $t$.'' 

The two key concepts events and worldlines of relativity theory can easily be defined from $W$ as follows. The \emph{event} coordinatized by observer $o$ at spacetime location $\bar x=(x,y,z,t)$ is the set of those bodies that are coordinatized at $\bar x$ by $o$, i.e.,
\[
ev_o(\bar x) \stackrel{\textsf{\tiny def.}}{=} \{ b: W(o,b,\bar x) \}.
\] 
   
 The \emph{worldline} of body $b$ according to observer $o$ is the set of coordinate points where $b$ is coordinatized by $o$, i.e.,
\[
wline_o(b) \stackrel{\textsf{\tiny def.}}{=} \{ \bar x: W(o,b,\bar x) \}.
\] 
For any two coordinate points $\bar x,\bar y\in Q^4$, let us use the following notations for  the \emph{spatial distance} and the \emph{time difference}:
\[
dist(\bar x,\bar y) \stackrel{\textsf{\tiny def.}}{=} \sqrt{(x_1-y_1)^2+(x_2-y_2)^2+(x_3-y_3)^2} \text{ and } time(\bar x,\bar y) \stackrel{\textsf{\tiny def.}}{=} |x_4-y_4|.\enskip \footnote{By axiom $AxEField$ (see below) we have a strong enough algebraic structure on the quantities to define subtraction and square root.}
\]

The \emph{speed} of body $b$ according to observer $i$ can be defined as follows, 
\[v_i(b)=v \stackrel{\textsf{\tiny def.}}{\iff} (\exists \bar x,\bar y\in wline_i(b))\big[space_i(\bar x,\bar y) = v\cdot time(\bar x,\bar y)\big] \]
if $wline_i(b)$ is a subset of a line, e.g., if $b$ is a light signal, $i$ is an inertial observer, and axioms $AxPh$ and $AxEField$ (below) are assumed. 

The \emph{worldview transformation} between observer $o$ and $o'$ connects coordinate points $\bar x$ and $\bar x'$ iff the event coordinatized by $o$ at $\bar x$ is the same as the one coordinatized by $o'$ at $\bar x'$, i.e., 
\[
w_{oo'}(\bar x) = \bar x' \stackrel{\textsf{\tiny def.}}{\iff} ev_o(\bar x) = ev_{o'}(\bar x').  
\]
By the above definition worldview transformations are only binary relations but after assuming some axioms they become transformations, see, e.g., Theorem \ref{thm-apoi}.

\subsection{Axioms for special relativity}

Einstein's two original postulates immediately imply that the speed of light is the same for every inertial observer because this is true for one of the inertial observers by the second postulate, and there is no distinguished inertial observer by the first postulate. This property of light signals is basically the only nontrivial assumption we need to capture the kinematics of special relativity. 

\underline{$AxPh$}: For every inertial observer $i$, there is a finite speed $c_i$ such that all light signals move with speed $c_i$ according to $i$, and it is possible to send out a light signal with this speed $c_i$ in every direction everywhere, i.e.,\footnote{It is more natural to assume that the value $c$ is the same for all inertial observers in a separate axiom, see axioms $Ax(c_i=c_j)$ and $Ax(c=1)$ on p.~\pageref{p-c}.}
\begin{multline*}
(\forall i\in IOb)(\exists c_i\in Q)(\forall\bar x,\bar y\in Q^4)\Big[(\exists p\in Ph) \big[\bar x,\bar y\in wline_i(p)\big] \\ \iff dist(\bar x, \bar y) = c_i\cdot time(\bar x, \bar y)\Big].\;\footnotemark
\end{multline*}
Not just to prove their usual properties, but even to be able to define the concepts of spatial distance and time difference, we need some assumptions about the quantities. Therefore, we assume some algebraic properties of real numbers.\footnotetext{This formula literally says that for any inertial observer $i$ there is a quantity $c_i$ such that a potential light signal $p$ can be located in coordinate points $\bar x$ and $\bar y$ according to $i$ if and only if the corresponding speed is $c_i$.}

\underline{$AxEField$}: $(Q, +, \cdot, <)$ is an Euclidean field.\footnote{That is $(Q, +, \cdot)$ is a field, $(Q,<)$ is a linearly ordered set, $x<y\implies x+z< y+z$ and $0<x, 0< y \implies 0< x\cdot y$, and every positive number has a square root, i.e., $(\forall x>0)\exists y \big[x=y^2\big]$.    }

Also $AxPh$ has its intended meaning only if we assume that inertial observers coordinatize the same outside reality. Therefore, we also need the following assumption. 
 
\underline{$AxEv$}: Inertial observers coordinatize the same events, i.e., 
\[
(\forall i,i'\in IOb)(\forall\bar x\in Q^4)(\exists\bar x'\in Q^4)\big[ ev_i(\bar x) = ev_{i'}(\bar x')\big].
\]

Finally, to make it easier speaking about the motion of coordinate systems in the usual way by referring to the image of their time axes in other coordinate systems, we assume that the worldline of every inertial observer is the time axis in her/his own coordinate system. 

\underline{$AxSelf$}: Inertial observers are stationary in the origin of their own coordinate systems, i.e.,
\[
(\forall i \in IOb)\forall xyzt \big[(x,y,z,t) \in wline_i(i) \iff x=y=z=0\big].
\]

The four simple axioms above are enough to capture special relativity in a qualitative way. So let us introduce $SpecRel_0$ as the axiom system containing the four axioms above:
\[SpecRel_0 \stackrel{\textsf{\tiny def.}}{=}  AxPh +  AxEField  + AxEv + AxSelf.\]

\subsection{Consequences of $SpecRel_0$}

The first thing the reader might spot is that there is no axiom in $SpecRel_0$ stating that inertial observers move uniformly. This is so because this statement follows from the axioms of $SpecRel_0$:

\begin{thm}\footnote{By $AxSelf$ this theorem is a direct consequence of Theorem~\ref{thm-apoi} below.}${}$\\
\vspace{-24pt}
\begin{multline*}SpecRel_0 \implies (\forall i,j\in IOb)(\forall \bar x,\bar y,\bar z\in wline_i(j))\big[\bar x\neq \bar y \\\implies \exists \lambda [\lambda \cdot (\bar x -\bar y)= \bar z -\bar y] \big]
\end{multline*}
\end{thm} 

We made no explicit restriction on the speed of the inertial observers, but $SpecRel_0$ implies that they cannot move faster than the speed of light with respect to one another. 

\begin{thm}\label{thm-noftl}
$SpecRel_0 \implies (\forall i,j\in IOb)(\forall p\in Ph) \big[v_i(j)<v_i(p)\big]$ 
\end{thm}

A direct proof of Theorem \ref{thm-noftl} from slightly stronger axioms can be found, for example, in Andr\'eka et.\ al.\ (2012a).

By Theorem \ref{thm-async} below $SpecRel_0$ contradicts the Newtonian notion of absolute time and replaces it with an observer dependent one. For a formal statement and a direct proof from slightly stronger axioms, see, e.g., Andr\'eka et.\ al.\ (2007, Thm.\ 11.4, pp.~626-630).  

\begin{thm}\label{thm-async}  $SpecRel_0 \implies$ ``For any two relatively moving inertial observers, there are pairs of events which are simultaneous for one of the observers but not for the other.''
\end{thm}

Time dilation and length contraction are two famous predictions of special relativity. Even though $SpecRel_0$ does not imply the exact rates of the dilation and the contraction, it  predicts these effects qualitatively in the following way. 
\begin{thm} \leavevmode \vspace{-\baselineskip} 
\begin{enumerate}
\item $SpecRel_0 \implies$ ``At least one of two relatively moving inertial observers sees that the other's clocks are slowed down.''
\item $SpecRel_0 \implies$ ``At least one of two relatively moving inertial spaceships shrinks according to the other.''
\end{enumerate}
\end{thm}

See Andr\'eka et.\ al.\ (2002, Sec.\ 2.5) for precise formulations of these statements. 

All the theorems above can also be derived from Theorem \ref{thm-apoi} below. 
To state this theorem let us recall that, a map $P:Q^4\rightarrow Q^4$ is a \emph{Poincar\'e transformation} if it is an affine bijection with the following property
\begin{equation*}
time(\bar x,\bar y)^2-dist(\bar x,\bar y)^2=time(\bar x',\bar y')^2-dist(\bar x',\bar y')^2
\end{equation*}
for all $\bar x,\bar y,\bar x',\bar y'\in Q^4$ for which $P(\bar x)=\bar x'$ and
$P(\bar y)=\bar y'$.

Since Poincar\'e transformations are the standard coordinate transformations in special relativity, Theorem \ref{thm-apoi} below basically says that $SpecRel_0$ implies these standard coordinate transformations up to changing units of measurement.\footnote{Permuting the coordinates  by a field automorphism can also be considered as an unusual way of changing the units of measurement if the underlying field has a nontrivial automorphism at all. For example, the field of rational numbers or the field of real numbers does not have a nontrivial automorphism.} 

\begin{thm}\label{thm-apoi}
$SpecRel_0 \implies$ ``The worldview transformations between inertial observers are Poincar\'e transformations up to changing the units of measurement and permuting all coordinates by a field automorphism.''\,\footnote{That is, the coordinate transformation between inertial observers $i$ and $i'$ becomes a Poincar\'e transformation after multiplying all the space coordinates by a positive number, the time coordinates with a possibly different number, and maybe also transforming the coordinate system of $i$ by a mapping $(x,y,z,t)$ to $(\phi(x),\phi(y),\phi(z),\phi(t))$, where $\phi$ is an automorphism of $(Q,+,\cdot)$ and doing the same with the coordinates of $i'$ but with possibly different numbers and automorphism.} 
\end{thm}

Theorem \ref{thm-apoi} follows immediately from Andr\'eka et.~al.\ (2012b, Theorem 7.8) by (Ibid Proposition 7.19).

\section{SYMMETRY AXIOMS IN SPECIAL RELATIVITY}

A great many symmetry axioms were formulated and investigated in Andr\'eka et.\ al.\ (2002, Sec.\ 2.8 and Sec.\ 3.9). Here we highlight only two of them, both of which are about harmonizing the units of measurement of different inertial observers. 

So how can we test if two different inertial observers use the same units of measurement or not?

\subsection{Symmetry of space}

To compare the units of two different inertial observers measuring spatial distances, we can to ask them to determine the distances of events which are simultaneous for both of them. If they get the same distances for these events, then they use the same units. Otherwise, the one who uses the smaller units gets bigger numbers for the distances. In the next axiom, we can use this idea to formulate  that different inertial observers use the same units to measure spatial distances.  

\underline{$AxSymDist$:}  Inertial observers agree as for the spatial distance between events if these events are simultaneous for both of them, i.e.,
\begin{multline*}
(\forall i,i'\in IOb)(\forall \bar x,\bar y,\bar x',\bar y'\in Q^4)\Big[ev_i(\bar x)=ev_{i'}(\bar x'), ev_i(\bar y)=ev_{i'}(\bar y'), \\ time (\bar x,\bar y)=0, time (\bar x',\bar y')=0  \implies dist(\bar x,\bar y)=dist (\bar x',\bar y')\Big].   
\end{multline*}

\subsection{Symmetry of time}

Another idea to compare the units of measurement of different inertial observers is to ask them to compare the ticking rate of each others' clocks. If they see each others' clock behaving the same way, e.g., slowed down with the same rate, then they use the same units to measure time differences. Otherwise, the one who uses the smaller unit considers the other's clock slowed down more (he would consider the other's clocks slow even if they are stationary with respect to each other).  The next axiom uses this idea to formulate that different inertial observers use the same units to measure time differences.  

\underline{$AxSymTime$:}  Any two inertial observers see each others' clocks slowed down with the same rate, i.e.,
\[
(\forall i,i'\in IOb)\forall t\Big[time\big(w_{ii'}(t\cdot \bar 1_t),w_{ii'}(\bar o)\big)=time\big(w_{i'i}(t\cdot \bar 1_t),w_{i'i}(\bar o)\big)\Big],
\]
where $\bar o = (0,0,0,0)$ and $\bar 1_t = (0,0,0,1)$. 

\subsection{Connections between the symmetry axioms}

The value of the speed of light depends on the units of measurement inertial observers choose to measure space and time. After having $AxPh$ assumed, assuming that the speed of light is 1 only means that inertial observers measure spatial distance in the units corresponding to their time unit. For example, if somebody measures time in years he only has to measure distance in light-years to have 1 for the speed of light, etc. 

If $AxPh$ and $AxEField$ are assumed, then we can ensure that the speed of light is the same for different inertial observers by the following axiom. 

\underline{$Ax(c_i=c_j)$}: Inertial observers use units to measure time differences and spatial distances such that the speed of light is the same for all of them, i.e., \label{p-c}
\[
(\forall i,j\in IOb)(\forall p\in Ph)\big[v_i(p)=v_j(p)\big].
\] 
Likewise, if $AxPh$ and $AxEField$ are assumed, then we can set the speed of light to 1 for all the inertial observers by the following axiom. 

\underline{$Ax(c=1)$}: Inertial observers use units to measure time differences and spatial distances such that the speed of light is one for all of them, i.e.,
\[
(\forall i\in IOb)(\forall p\in Ph)\big[v_i(p)=1\big].
\] 

Originally, $Ax(c=1)$ was part of $AxPh$, see, e.g., Andr\'eka et.\ al.\ (2002, AxE on p.~51) or Andr\'eka et.\ al.\ (2007, p.~621). Later it was moved to the corresponding symmetry axiom because it is strongly related to the units of measurement chosen by inertial observers, see, e.g., Andr\'eka et.\ al.\ (2012a). Here we have introduced $Ax(c=1)$ as an axiom in its own right because it is easier to understand the roles and intuitive meanings of $AxSymDist$ and $Ax(c=1)$ if they are separated. 

Let us introduce axiom system $SpecRel$ as follows: 
\[
SpecRel\stackrel{\textsf{\tiny def.}}{=}SpecRel_0+AxSymDist+Ax(c=1).
\]

\begin{thm}\label{thm-cons}
$SpecRel$ is a consistent axiom system. Moreover, there is a model of $SpecRel$ where for every inertial observer $i$ and Poincar\'e transformation $P$ there is an inertial observer $j$ such that $w_{ij}=P$.
\end{thm}

See Andr\'eka et.\ al.\ (2002, Sec.\ 3.6) and Andr\'eka et.\ al.\ (2007, Sec.\ 2.5) for constructions proving theorem~\ref{thm-cons}. 

By Proposition \ref{prop-ineq} below, $AxSymTime$ and $AxSymDist$ are independent from the rest of the axioms of $SpecRel$ and they are not equivalent if only $SpecRel_0$ is assumed. 

\begin{prop} \label{prop-ineq} \leavevmode \vspace{-\baselineskip} 
\begin{enumerate}
\item $SpecRel_0 + Ax(c=1)\centernot\implies AxSymTime$

\item $SpecRel_0 + Ax(c=1) \centernot\implies AxSymDist$

\item $SpecRel_0 + AxSymTime \centernot\implies AxSymDist$

\item $SpecRel_0 + AxSymDist \centernot\implies AxSymTime$
\end{enumerate}
\end{prop}

\begin{proof}[On the proof.]
All of the items of this proposition are proved by constructing appropriate models. We can start from any model of $SpecRel$ containing enough (at least two) observers. For example, the one whose existence is stated by Theorem~\ref{thm-cons}.  

To prove items 1.\ and 2., we should construct a model of $SpecRel_0 + Ax(c=1)$, where $AxSymTime$ and $AxSymDist$ do not hold. If we modify an inertial observer's coordinate system in a model of $SpecRel_0 + Ax(c=1)$ by multiplying all the coordinates with the same positive number, we get another model of $SpecRel_0 + Ax(c=1)$. Using this modification it is easy to make $AxSymTime$ and $AxSymDist$ invalid in the models of $SpecRel$ without changing the validity of the rest of the axioms.  

To prove items 3.\ and 4., we should make one of $AxSymTime$ and $AxSymDist$ invalid without changing the validity of the other and the axioms of $SpecRel_0$. This can be done easily by scaling only the time or only the space coordinates of the inertial observers.\footnote{Of course, this construction also makes axiom $Ax(c=1)$ invalid by Proposition \ref{thm-eq}.}   
\end{proof}

By Proposition \ref{thm-eq} below, $AxSymTime$ and $AxSymDist$ are equivalent assuming $Ax(c_i=c_j)$ and $SpecRel_0$. 

\begin{prop}\label{thm-eq} 
$SpecRel_0 + Ax(c_i=c_j)\implies (AxSymTime \iff AxSymDist)$
\end{prop} 

\begin{proof}[On the proof] \vspace{-\baselineskip} 
If $Ax(c=1)$ is also satisfied, both $AxSymTime$ and $AxSymDist$ are equivalent to the statement that the worldview transformations are Poincar\'e transformations. 
The other cases can easily be reduced to the $c=1$ case by multiplying all the observer's time coordinates by factor $c$ (i.e., the speed of light that is the same for all inertial observers by axiom $Ax(c_i=c_j)$). 
\end{proof}

By Proposition \ref{prop-c} below, $AxSymTime$ and $AxSymDist$ imply that different inertial observers agree on the exact value of the speed of light if $SpecRel_0$ is assumed. 

\begin{prop}\label{prop-c}
$SpecRel_0 + AxSymTime + AxSymDist \implies Ax(c_i=c_j)$
\end{prop}

\begin{proof}[On the proof.]
Using Theorem \ref{thm-apoi}, it is easy to prove this proposition. Axioms $AxSymTime$ and $AxSymDist$ fix both the time and space units of measurement of inertial observers. Without changing these units what remains from the possible worldview transformations characterized by Theorem \ref{thm-apoi} leaves no flexibility to have inertial observers with different values for the speed of light. 
\end{proof}

\subsection{The role of the symmetry axioms in certain predictions of relativity}

By Theorem~\ref{thm-par} below, $SpecRel$ captures time dilation and length contraction predictions of special relativity even quantitatively. 

\newpage 
\begin{thm}\label{thm-par} \leavevmode \vspace{-\baselineskip} 
  \begin{enumerate}
  \item $SpecRel \implies$ ``Relatively moving inertial observers see that each others' clocks are slowed down exactly by the Lorentzian contraction factor.''
  \item $SpecRel \implies$ ``Relatively moving inertial observers see each others' spatial distances shrinking exactly by the Lorentzian contraction factor in the direction of motion.''
  \end{enumerate}
\end{thm}

See Andr\'eka et.\ al.\ (2007, Thm. 11.6, pp.~631-635) for precise formalization of these statements and direct proofs from the axioms.  

The so called twin paradox is the prediction of relativity stating that between two distinct meeting points inertial observers always measure more time than non-inertial ones do. Since it considers only inertial observers, $SpecRel$ is not strong enough to investigate the twin paradox in this form.\footnote{See Madar\'asz et.\ al.\ (2006) and Sz\'ekely (2009, Sec.\ 7) for an axiomatic investigation of the twin paradox in this form within a theory of accelerated observers.} However, we can simulate the accelerated twin by two inertial observers: a leaving one and a returning one. With this standard trick we can introduce an inertial version of the twin paradox where in the simulated twin paradox situations (cf.\ Figure \ref{fig-twp}) the stay-at-home inertial twin always measures more time than his leaving and returning inertial sisters together. 

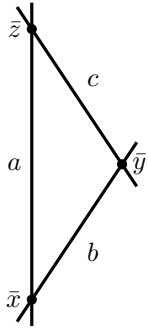
\begin{figure}[!htb]
\begin{center}
\begin{tikzpicture}[scale=0.6]
\draw[white] (-1,0) rectangle (3,7);
\draw[very thick,shorten <=-10,shorten >=-10] (0,0) node[left] {$\bar x$} to node[below right] {$b$}(2,3) node[right] {$\bar y$};
\draw[very thick,shorten <=-10,shorten >=-10] (2,3) to node[above right] {$c$} (0,6) node[left] {$\bar z$};
\draw[very thick,shorten <=-10,shorten >=-10] (0,0) to node[left]{$a$} (0,6);
\draw[fill] (0,0) circle (0.1);
\draw[fill] (2,3) circle (0.1);
\draw[fill] (0,6) circle (0.1);
\end{tikzpicture}
\end{center}
\caption{\label{fig-twp} Illustration for the formula $TwP$}
\end{figure}

$TwP$: In the situations depicted by Figure \ref{fig-twp} inertial observer $a$ measures more time between the events at $\bar x$ and $\bar y$ than inertial observers $b$ and $c$  together, i.e.,  
\begin{multline*}
(\forall a,b,c\in IOb)(\forall \bar x,\bar z \in wline_a(a))(\forall \bar y\not\in wline_a(a))\\\Big[\bar x\in wline_a(b),\bar y\in wline_a(b)\cap wline_a(c), \bar z\in wline_a(c) \\ \implies time(\bar x,\bar z) > time\big(w_{ab}(\bar x),w_{ab}(\bar y)\big)+ time\big(w_{ac}(\bar y),w_{ac}(\bar z)\big)\Big]
\end{multline*}

By Theorem~\ref{thm-twp} below, $SpecRel$ implies $TwP$ but not without the symmetry axioms.  

\begin{thm}\label{thm-twp}\leavevmode \vspace{-\baselineskip} 
\begin{enumerate}
\item $SpecRel_0 + Ax(c=1) \centernot\implies TwP$
\item $SpecRel_0 + Ax(c=1) \implies (AxSymTime \implies TwP)$
\end{enumerate}
\end{thm}

In view of Theorem \ref{thm-twp}, whether $TwP$ also implies $AxSymTime$ or not is a natural question asked by Andr\'eka et.\ al.\ (2002, Question 4.2.17). By the intuitive insight of this paper about the role of axiom $AxSymTime$, it is not surprising that $TwP$ cannot fulfill that role and hence it does not imply $AxSymTime$.  

\begin{prop}\label{prop-nosym}
$SpecRel_0 + Ax(c=1) + TwP \centernot\implies AxSymTime$
\end{prop}

See Sz\'ekely (2010) and Sz\'ekely (2009 Sec.\ 4) for proof of Theorem~\ref{thm-twp} and Proposition~\ref{prop-nosym} and a detailed investigation of $TwP$, e.g., a geometrical characterization of the models of $SpecRel_0 + Ax(c=1)$ where $TwP$ holds. 

\section{SYMMETRY AXIOMS IN GENERAL RELATIVITY}

Both in the special and the general theories of relativity, this kind of axiomatic approach to symmetries based on comparing the units of measurement of different observers is original to the Andr\'eka--N\'emeti team. To discuss the general relativistic counterparts of symmetry axioms $AxSymTime$ and $AxSymDist$, let us  see first how general relativity emerges in the axiomatic framework of the team.  

The transition from special relativity to general relativity in terms of axioms basically goes by assuming the localized versions of the same axioms, see, e.g., Andr\'eka et.\ al.\ (2012a). The corresponding axiom system is called $GenRel$. Theorem~\ref{thm-comp} below says that the axiom system $GenRel$ captures general relativity well. For more details, e.g., precise statement, proof, and refinements, see, e.g., Andr\'eka et.\ al.\ (2013).  

\begin{thm}[Completeness Theorem]\label{thm-comp}
$GenRel$ is complete with respect to the standard models of general relativity, i.e., to Lorentzian manifolds. 
\end{thm}

How are the axioms of $GenRel$ related to the axioms of $SpecRel$?

For every axiom $Ax$ of $SpecRel$, there is an axiom $Ax^-$ in $GenRel$, such that $Ax^-$ captures the same idea locally for arbitrary observers as $Ax$ does globally for inertial ones. The only exception is axiom $\mathit{AxDiff}$ of $GenRel$ because it localizes not an axiom but the theorem of $SpecRel$ stating that the worldview transformations between inertial observers are affine transformations.  

In $GenRel$, the localized version of axiom $Ax(c=1)$ is included in $AxPh^-$. Both $AxSymDist$ and $AxSymTime$ can be localized, but the localization of $AxSymTime$ is simpler, see Sz\'ekely (2009, pp.~96-97). $AxSymTime^-$ intuitively says the following. 

\underline{$AxSymTime^-$}: Any two observers see each others' clocks behaving in the same way at an event of meeting.

The standard approach to symmetries in general relativity is based on local diffeomorphisms preserving some geometrical notions, such as the metric or geodesics, see, e.g., Hall (2004, Sec.\ 10). Even though it is a natural open question how these standard concepts of symmetries can be captured within the first-order logic framework of the Andr\'eka--N\'emeti team, here we concentrate on another interesting question, namely on what role may the symmetry axioms play in $GenRel$.

In $SpecRel$ the only flexibility we get by leaving out the symmetry axioms is that different inertial observers may use different units of measurements, see Theorem~\ref{thm-apoi}. Therefore, it is natural to conjecture that in $GenRel$ the localized version of the symmetry axioms will have similar roles, i.e., without them the rest of the axioms will still capture the standard models of general relativity up to conformal equivalence (i.e., up to changing of units of measurements locally).\footnote{Since $AxPh^-$ contains the localized version of $Ax(c=1)$, the units for measuring time and space remain intertwined even if we omit $AxSymTime^-$. The nonstandard unit changing transformations induced by field automorphisms in Theorem~\ref{thm-apoi} are conjectured to be ruled out by $\mathit{AxDiff}$.}

\begin{conj}\label{conj}
If we leave out $AxSymTime^-$ from $GenRel$ it will still capture Lorentzian manifolds up to conformal equivalence. 
\end{conj}

That is, even though no unique Lorentzian metric can be defined without axiom $AxSymTime^-$, the manifold of events and the tangent space is definable the same way  as in Andr\'eka et.\ al.\ (2013) and  probably there will be some definable geometric object in the tangent space that can capture the conformal equivalence classes of Lorentzian manifolds.

\section{CONCLUDING REMARKS}

We have seen that symmetry axioms $AxSymTime$ and $AxSymDist$ ensure that different inertial observers use the same units of measurement to measure time and space, respectively. This insight it helps to understand why these symmetry axioms are needed to get the exact Lorentzian contraction factors and why $SpecRel_0$ implies time dilation and length contraction only qualitatively. Also after understanding the exact roles of these symmetry axioms, it becomes clearer why they are needed to prove the twin paradox ($TwP$), and why $TwP$ is not strong enough to imply them. 

It is natural to conjecture that the localized version of these symmetry axioms will have a similar role in general relativity, but to check this conjecture requires further investigations. 

\section{REFERENCES}

\noindent
Andr{\'e}ka, H., Madar{\'a}sz, J.~X., and  N{\'e}meti, I., with contributions from {A}.~{A}ndai, {G}.~{S}{\'a}gi, {I}.~{S}ain, and {C}s.~{T}{\H o}ke.\ (2002)   
\textit{On the logical structure of relativity theories}. E-book, {A}lfr{\'e}d {R}{\'e}nyi {I}nstitute  of {M}athematics, 1312pp. 

\noindent
Andr{\'e}ka, H., Madar{\'a}sz, J.~X., and  N{\'e}meti, I.\ (2006).   Logical axiomatizations of space-time. {S}amples from the literature.
 In A.~Pr{\'e}kopa and E.~Moln{\'a}r, editors, \textit{Non-{E}uclidean geometries}, pages 155--185. Springer-Verlag, New York.

\noindent
Andr{\'e}ka H.,  Madar{\'a}sz, J.~X., and N{\'e}meti, I.\ (2007),  Logic of space-time and relativity theory,   In: {M}.~{A}iello, {I}.~{P}ratt-{H}artmann, and {J}.~van {B}enthem, eds.,  \textit{Handbook of spatial logics}, pp.~607--711. {S}pringer-Verlag, {D}ordrecht.

\noindent
Andr{\'e}ka, H., Madar\'asz, J.~X., N{\'e}meti, I., and Sz\'ekely, G.\ (2012a) 
A logic road from special relativity to general relativity, \textit{Synthese}, 186, no.~3,633-469.  

\noindent
Andr\'eka, H.,  Madar\'asz, J.X., N\'emeti I., and Sz\'ekely, G.\ (2012b)
What are the numbers in which spacetime? arXiv:1204.1350. 

\noindent
Andr\'eka, H.,  Madar\'asz, J.X., N\'emeti I., Sz\'ekely, G.\ (2013)  An Axiom System for General Relativity Complete with respect to Lorentzian Manifolds.   	arXiv:1310.1475. 

\noindent
Andr\' eka, H., and N\' emeti, I.\ (2014), Comparing theories: the dynamics
of changing vocabulary. In: \textit{Johan V. A. K. van Benthem on logical and informational dynamics.} A. Baltag and S. Smets, eds., Springer Series Outstanding
contributions to logic Vol 5, Springer Verlag,  143-172.

\noindent
Ax, J.\ (1978) The elementary foundations of spacetime.
\textit{Foundations of Physics}, 8, no.\ 7-8,507-546.

\noindent
Benda, T.\ (2008) A formal construction of the spacetime manifold.
\textit{Journal of Philosophical Logic}, 37, no.\ 5, 441-478.

\noindent
Benda, T.\ (2015) An axiomatic foundation of relativistic spacetime.
\textit{Synthese}, 192, no.\ 7, 2009-2024.

\noindent
Barrett, T. W., and Halvorson, H.\ (2015), Morita equivalence. 
arXiv:1506.04675. 

\noindent
Einstein, A.\ (1905) On the Electrodynamics of Moving Bodies. \textit{Annalen der Physik} 17, 891-921.

\noindent
Friend, M.\ (2015) On the epistemological significance of the hungarian project. \textit{Synthese}, 192, no.\ 7,  2035-2051.

\noindent
Goldblatt, R.\ (1987) \textit{Orthogonality and spacetime geometry}. Springer-Verlag, ix + 190pp.

\noindent
Govindarajalulu, N.\ S., Bringsjord S., and Taylor J.\ (2015)
Proof verification and proof discovery for relativity.
\textit{Synthese}, 192, no.\ 7, 2077-2094.

\noindent
Guts, A.\ K.\ (1982) 
The axiomatic theory of relativity.
\textit{Russian Mathematical Surveys}, 37, no.\ 2, 41-89.

\noindent
Hall, G.\ S.\ (2004), \textit{Symmetries and curvature structure in general relativity}. World Scientific Lecture Notes in Physics: Volume 46, World Scientific, x + 430pp.

\noindent
Madar\'asz, X.~J.\ (2002), \textit{Logic and Relativity (in the light of definability theory)}.  [Ph.D. Dissertation], {B}udapest: {E}{\"o}tv{\"o}s {L}or{\'a}nd University, 444 pp. 

\noindent
Madar\'asz, J.X., N\'emeti I., and Sz\'ekely, G.\ (2006) Twin Paradox and the logical foundation of relativity theory. \textit{Foundations of Physics} 36, no. 5, 681-714. 

\noindent
Mundy, B.\ (1986) Optical axiomatization of Minkowski space-time geometry.
\textit{Philosophy of Science}, 53, no.\ 1, 1-30.

\noindent
Pambuccian, V.\ (2007) 
Alexandrov-Zeeman type theorems expressed in terms of definability.
\textit{Aequationes Mathematicae}, 74, no.\ 3, 249-261.

\noindent
Rosenstock, S., Barrett, T. W., and Weatherall, J. O.\ (2015). On Einstein
algebras and relativistic spacetimes. arXiv:1506.00124.

\noindent
Schutz,  J.\ W.\ (1981) 
An axiomatic system for Minkowski space-time.
\textit{Journal of Mathematical Physics}, 22, no.\ 2, 293-302.

\noindent
Stannett, M., and N\'emeti I., (2014)
Using Isabelle/HOL to Verify First-Order Relativity Theory.
\textit{Journal of Automated Reasoning}, 52, no.\ 4, 361-378.  

\noindent
Szab\'o, L.\ E.\ (2009) Empirical Foundation of Space and Time.  In: M. Su\'arez, M. Dorato and M. R\'dei, eds., EPSA07: \textit{Launch of the European Philosophy of Science Association}, Springer, 251-266. 

\noindent
Sz\'ekely, G.\ (2009), \textit{First-Order Logic Investigation of Relativity Theory with an Emphasis on Accelerated Observers}.  [Ph.D. Dissertation] {B}udapest: {E}{\"o}tv{\"o}s {L}or{\'a}nd University, 152 pp.

\noindent
Sz\'ekely G.\ (2010) A Geometrical Characterization of the Twin Paradox and its Variants. \textit{Studia Logica} 95, no.~1-2, 161-182. 

\noindent
Weatherall, J. O.\ (2014), Are Newtonian gravitation and geometrized Newtonian gravitation theoretically equivalent?  arXiv:1411.5757.

\end{document}